\documentclass[preprint]{elsarticle}
\usepackage[T1]{fontenc}
\usepackage{pgf}
\usepackage[latin1]{inputenc}
\usepackage[english]{babel}
\usepackage{amsmath,enumerate}
\usepackage{amssymb}
\usepackage{amscd}
\usepackage{tabularx}
\usepackage{vmargin}
\usepackage{url}
\usepackage{stmaryrd}

\begin{document}

\begin{frontmatter}
\title{Strong edge-colouring of sparse planar graphs}
\author[LaBRI]{Julien Bensmail}
\author[Oxford]{Ararat Harutyunyan}
\author[LaBRI]{Herv\'e Hocquard}
\author[LIP]{Petru Valicov\corref{mycorrespondingauthor}}
\cortext[mycorrespondingauthor]{Corresponding author}
\ead{petru.valicov@ens-lyon.fr}

\address[LaBRI]{LaBRI (Universit\'e de Bordeaux), 351 cours de la Lib\'eration, 33405 Talence Cedex, France}
\address[Oxford]{Mathematical Institute, University of Oxford, United Kingdom}
\address[LIP]{\'Ecole Normale Sup\'erieure de Lyon, LIP, \'Equipe MC2, 46, all\'ee d'Italie, 69342 Lyon Cedex 07, France}

\begin{abstract}
A strong edge-colouring of a graph is a proper edge-colouring where each colour class induces a matching. It is known that every planar graph with maximum degree $\Delta$ has a strong edge-colouring with at most $4\Delta+4$ colours. We show that $3\Delta+1$ colours suffice if the graph has girth 6, and $4\Delta$ colours suffice if $\Delta\geq 7$ or the girth is at least 5. In the last part of the paper, we raise some questions related to a long-standing conjecture of Vizing on proper edge-colouring of planar graphs.
\end{abstract}


\end{frontmatter}

\newenvironment{proof}{\par \noindent \textbf{Proof} \\}{\hfill$\Box$}

\newtheorem{theorem}{Theorem}
\newtheorem{corollary}[theorem]{Corollary}
\newtheorem{question}{Question}
\newtheorem{proposition}{Proposition}
\newtheorem{conjecture}[theorem]{Conjecture}
\newtheorem{remark}[theorem]{Remark}
\newtheorem{claim}{Claim}

\section{Introduction}
A \emph{proper edge-colouring} of a graph $G=(V,E)$ is an assignment of colours to the edges of the graph such that two adjacent edges do not use the same colour. We use the standard notation, $\chi'(G)$, to denote the chromatic index of $G$.
A \emph{strong edge-colouring} (called also distance 2 edge-colouring) of a graph $G$ is a proper edge-colouring of $G$, such that the every set of edges using the same colour induces a matching.
We denote by $\chi'_s(G)$ the \emph{strong chromatic index} of $G$ which is the smallest integer $k$ such that $G$ can be strongly edge-coloured with $k$ colours. Strong edge-colouring
has been studied extensively in the literature by different authors (see~\cite{E88,EN89,MR97,FGST90,A92,HQT93,HMRV13,BI13, CMPR13,HLSS13}).

The \emph{girth} of a graph $G$ is the length of a shortest cycle in $G$. We denote by $\Delta$ the maximum degree of a graph.

Perhaps the most challenging question for strong edge-colouring is the following conjecture:
\begin{conjecture}[Erd\H{o}s and Ne\v{s}et\v{r}il~\cite{EN89}]
\label{conj:erdos_nesetril}
For every graph $G$, $\chi'_s(G)\leq \frac{5}{4}\Delta^2$ for $\Delta$ even and $\frac{1}{4}(5\Delta^2-2\Delta+1)$ for $\Delta$ odd.
\end{conjecture}

Andersen~\cite{A92} and Hor{\'a}k \emph{et al.}~\cite{HQT93} showed this conjecture for the case when $\Delta=3$. When $\Delta$ is large enough, Molloy and Reed showed that $\chi'_s(G)\leq 1.998\Delta^2$~\cite{MR97}.

In this note, we study the strong chromatic index of planar graphs. The work in this area started with the paper of Faudree \emph{et al.}~\cite{FGST90}, who proved the following theorem.

\begin{theorem}[Faudree \emph{et al.}~\cite{FGST90}]
\label{thm:faudree_planar}
If $G$ is a planar graph then $\chi'_s(G)\leq 4\Delta+4$, for $\Delta\ge 3$.
\end{theorem}
The proof of Theorem~\ref{thm:faudree_planar} uses the Four Colour Theorem.
The authors also provided a construction of planar graphs of girth 4 which satisfy $\chi'_s(G)=4\Delta-4$. Hence, the bound of Theorem~\ref{thm:faudree_planar} is optimal up to an additive constant.

The same authors also conjectured that for $\Delta=3$ the bound can be improved.
\begin{conjecture}[\label{conj:faudree}Faudree \emph{et al.}~\cite{FGST90}]
If $G$ is a planar subcubic graph then $\chi'_s(G)\leq 9$.
\end{conjecture}

Hocquard \emph{et al.} obtained the following weakening of Conjecture \ref{conj:faudree}.

\begin{theorem}[\label{thm:subcubic}Hocquard \emph{et al.}~\cite{HMRV13}]
If $G$ is a planar graph with $\Delta\leq 3$ containing neither induced 4-cycles, nor induced
5-cycles, then $\chi'_s(G)\leq 9$.
\end{theorem}

An interesting question is to see how the strong chromatic index behaves for sparse planar graphs. For instance, when the girth is large enough the strong chromatic index decreases to the near optimal lower bound, as showed in the following theorems:

\begin{theorem}[\label{thm:large_girth_Borodin}Borodin and Ivanova~\cite{BI13}]
If $G$ is a planar graph with maximum degree $\Delta\geq 3$ and girth $g\geq 40\lfloor\frac{\Delta}{2}\rfloor$, then $\chi'_s(G)\leq 2\Delta-1$.
\end{theorem}

Recently this result was improved for $\Delta \geq 6$:

\begin{theorem}[\label{thm:large_girth_Chang}Chang \emph{et al.}~\cite{CMPR13}]
If $G$ is a planar graph with maximum degree $\Delta\geq 4$ and girth $g\geq 10\Delta+46$, then $\chi'_s(G)\leq 2\Delta-1$.
\end{theorem}

For smaller values of the girth, Hud\'ak \emph{et al.}~\cite{HLSS13} improved the bound in Theorem~\ref{thm:faudree_planar}.

\begin{theorem}[\label{thm:Hudak_et_al_girth6}Hud\'ak \emph{et al.}~\cite{HLSS13}]
If $G$ is a planar graph with girth $g\geq 6$ then $\chi'_s(G)\leq 3\Delta+6$.
\end{theorem}

Our main result in this paper improves the upper bound in Theorem~\ref{thm:Hudak_et_al_girth6}. In particular, we show the following.

\begin{theorem}\label{ABHV13}
If $G$ is a planar graph with girth $g\geq 6$ then $\chi'_s(G)\leq 3\Delta+1$.
\end{theorem}

Moreover, in Section~\ref{sec:girth5}, by a more careful analysis of the proof of
Theorem \ref{thm:faudree_planar} given in \cite{FGST90} and by using some results on proper edge-colouring, we obtain the following strengthening.

\begin{theorem} \label{thm:secondthm}
Let $G$ be a planar graph with maximum degree $\Delta$ and girth $g$. If $G$ satisfies one of the following conditions below, then $\chi_s'(G) \leq 4\Delta$

\begin{itemize}

\item $\Delta \geq 7$,
\item $\Delta \geq 5$ and $g\geq 4$,
\item $g\geq 5$.
\end{itemize}
\end{theorem}

%

Before proving our results we introduce some notation.

\paragraph{Notation}
Let $G$ be a graph. Let $d(v)$ denote the degree of a vertex $v$ in $G$. A vertex of degree $k$ is called a $k$-vertex. A $k^+$-vertex (respectively, $k^-$-vertex) is a vertex of degree at least $k$ (respectively, at most $k$). A $k_l$-vertex is a $k$-vertex adjacent to exactly $l$ $2$-vertices. A \emph{bad $2$-vertex} is a 2-vertex adjacent to another 2-vertex. When speaking about a vertex as a neighbour, same notations apply just by replacing the word "vertex" with "neighbour".
Two edges are at distance 1 if they share one of their ends and they are at distance 2 if they are not at distance 1 and there exists an edge adjacent to both of them.
We define $N_2[uv]$ as the set of edges at distance at most 2 from the edge $uv$ and $N_2(uv)=N_2[uv]-uv$. Given an edge-colouring of $G$, we denote by $SC(N_2(uv))$ ($SC(N_2[uv])$ respectively) the set of colours used by edges in $N_2(uv)$ ($N_2[uv]$ respectively). We denote by $N(v)$ the neighbourhood of the vertex $v$, {\it i.e.,} the set of its adjacent vertices. Finally, we use $\llbracket n \rrbracket$ to denote the set of integers $\{1,2,\ldots,n\}$.

\section{Proof of Theorem \ref{ABHV13}}

\subsection{Structural properties}

We proceed by contradiction. Let $H$ be a counterexample to the theorem that minimizes $|E(H)|+|V(H)|$. By minimality of $H$ we can assume that it is connected and that by Theorem~\ref{thm:subcubic} it has $\Delta(H)\geq 4$.

\begin{claim} \label{claim:girth6}
$H$ satisfies the following properties:
\begin{enumerate}
\item \label{claim1v}$H$ does not contain a $1$-vertex adjacent to a $4^-$-vertex.

\item \label{claim2v}$H$ does not contain a $2$-vertex adjacent to two $3^-$-vertices.

\item \label{claim423v}$H$ does not contain a $2$-vertex adjacent to a $3^-$-vertex and either a $4_2$-vertex or a $4_3$-vertex.

\item \label{claim4342v}$H$ does not contain a $2$-vertex adjacent to a $4_3$-vertex and to a $4_2$-vertex.

\item \label{claimk31v}If $k \ge 4$, then $H$ does not contain a $k$-vertex adjacent to $k-2$ 1-vertices; if the $k$-vertex is adjacent to $k-3$ 1-vertices, then it has no other $2^-$-neighbour.

\item \label{claimk2v}If $k \ge 4$, then $H$ does not contain a $k$-vertex adjacent to $k$ $2^-$-vertices.

\item \label{claimk12v}If $k \ge 5$, then $H$ does not contain a $k$-vertex $u$ with $N(u)=\{u_1,u_2,\ldots,u_{k-1},x\}$, such that each $u_i$ with $i\in\llbracket k-1 \rrbracket$ is a $2^-$-vertex and $u_1$ is either a 1-vertex or a $2$-vertex adjacent to either a $3^-$-vertex or a $4_3$-vertex.

\item \label{claimk22v}If $k \ge 5$, then $H$ does not contain a $k$-vertex adjacent to $k-2$ vertices of degree 2, $u_1,\ldots,u_{k-2}$, such that for $i\in \llbracket k-3 \rrbracket$, each $u_i$ is adjacent to either a $3^-$-vertex or a $4_3$-vertex.

\item \label{claima1v}If $k \ge 5$ and $1 \le \alpha \le k-4$, then $H$ does not contain a $k$-vertex adjacent to $\alpha$ $1$-vertices and to $k-2-\alpha$ vertices of degree 2, $u_1,\ldots,u_{k-2-\alpha}$, such that for $i \in\llbracket k-3-\alpha \rrbracket$ each $u_i$ is adjacent to either a $3^-$-vertex or a $4_3$-vertex.
\end{enumerate}
\end{claim}

\begin{proof}
Let $L$ be the set of colours $\llbracket 3\Delta+1 \rrbracket$. For each of the parts of the claim, we will suppose by contradiction that the described configuration exists in $H$. Then we will build a graph $H'$ from $H$ by removing a certain number of vertices and edges. By minimality of $H$  we will have $\chi'_s(H') \le 3\Delta+1$. Finally, for each of these cases, we will show a contradiction by showing how to extend a strong $(3\Delta+1)$-edge-colouring $\phi$ of $H'$ to a strong edge-colouring of $H$ without using an extra colour.
\begin{enumerate}
\item Suppose $H$ contains a $1$-vertex $u$ adjacent to a $4^-$-vertex $v$. Then let $H' = H - \{uv\}$. We can extend $\phi$ to $H$ by colouring $uv$ because $|L \setminus SC_\phi(N_2(uv))| \ge 1$.

\item Suppose $H$ contains a $2$-vertex $u$ adjacent to two $3^-$-vertices $v$ and $w$. Then let $H' = H - \{uv, uw\}$. Since $|L \setminus SC_\phi(N_2(uv))| \ge \Delta-1 \ge 3$ and $|L \setminus SC_\phi(N_2(uw))| \ge \Delta-1 \ge 3$, we can extend $\phi$ to $H$ by coulouring $uv$ and $uw$.

\item Suppose $H$ contains a $2$-vertex $u$ adjacent to a $3^-$-vertex $w$ and the other neighbour of $u$, say $v$, be either a $4_2$-vertex or $4_3$-vertex. We can assume that $v$ is a $4_2$-vertex since the proof will be the same when $v$ is a $4_3$-vertex. Then let $H' = H - \{uv, uw\}$. One can observe that $|L \setminus SC_\phi(N_2(uv))| \ge \Delta-3 \ge 1$ and $|L \setminus SC_\phi(N_2(uw))| \ge \Delta-2 \ge 2$. We can extend $\phi$ to $H$ by colouring $uv$ and $uw$ in this order.

\item Suppose $H$ contains a $2$-vertex $u$ adjacent to a $4_3$-vertex $v$ and to a $4_2$-vertex $w$. We assume that $N(w)=\{u,w_1,x,y\}$, $N(v)=\{u,v_1,v_2,z\}$, where $w_1$, $v_1$ and $v_2$ are $2$-vertices. Then take $H' = H - \{uv, uw\}$. In order to extend $\phi$ to $H$ we proceed as follows. We uncolour the edges $vv_1$ and $vv_2$. One can observe that $|L \setminus SC_\phi(N_2(uv))| \ge 2\Delta-4 \ge 4$ and $|L \setminus SC_\phi(N_2(uw))| \ge \Delta-2 \ge 2$. Hence, we colour $uv$ and $uw$. Observe now that $|L \setminus SC_\phi(N_2(vv_1))| \ge \Delta - 2 \ge 2$ and $|L \setminus SC_\phi(N_2(vv_2))| \ge \Delta - 2 \ge 2$. We can extend $\phi$ to $H$ by colouring $vv_1$ and $vv_2$.

\item Suppose $H$ contains a $k$-vertex $u$ adjacent to $k-3$ $1$-vertices $u_1, u_2,\ldots,u_{k-3}$  with $4 \le k \le \Delta$. Let $H' = H - \{uu_1\}$. We can extend $\phi$ to $H$ by colouring $uu_1$ which is possible because $|L \setminus SC_\phi(N_2(uu_1))| \ge \Delta -k + 3 \ge 3$.

\item Suppose $H$ contains a $k$-vertex $u$ adjacent to $k$ $2^-$-vertices $u_1, u_2,\ldots,u _{k}$ with $4 \le k \le \Delta$. Then let $H' = H - \{uu_1, uu_2,...,uu_k\}$. We extend $\phi$ to $H$ by colouring the edges $uu_1,...,uu_k$ in this order. Observe that for all $i \in \llbracket k \rrbracket$, $|L \setminus SC_\phi(N_2(uu_i))| \ge 2\Delta-2k+3 \ge 3$. Therefore, $\phi$ can be extended.

\item
Let $5 \le k \le \Delta$. Suppose $H$ contains a $k$-vertex $u$ with neighbours $u_1,u_2,\ldots,u_{k-1},x$ such that each $u_i$ with $i\in\llbracket k-1 \rrbracket$ is a $2^-$-vertex and $u_1$ is either a 1-vertex or a $2$-vertex adjacent to either a $3^-$-vertex or a $4_3$-vertex. Let $H' = H - \{u_1\}$. If $u_1$ is a 1-vertex then $\phi$ obviously can be extended to $H$. Therefore, $u_1$ is a 2-vertex. Let $v_1$ be the neighbour of $u_1$ other than $u$. By contradiction we assumed that $v_1$ is either a $3^-$-vertex or a $4_3$-vertex. In order to show how to extend $\phi$ to $H$, we consider two cases:
\begin{itemize}
\item If $v_1$ is a $3^-$-vertex, then $|L \setminus SC_\phi(N_2(uu_1))| \ge 2\Delta-2k+3 \ge 3$ and $|L \setminus SC_\phi(N_2(u_1v_1))| \ge \Delta-k+2 \ge 2$.
\item If $v_1$ is a $4_3$-vertex, then $|L \setminus SC_\phi(N_2(uu_1))| \ge 2\Delta-2k+2 \ge 2$ and $|L \setminus SC_\phi(N_2(u_1v_1))| \ge 2\Delta-k-2 \ge 2$.
\end{itemize}

Therefore, in both cases $\phi$ can be extended.

\item Let $u$ be a $k$-vertex in $H$ with $5 \le k \le \Delta$ such that there exists $k-2$ paths in $H$, $uu_jv_j$ with $j \in \llbracket k-2 \rrbracket$ and such that $d_H(u_j)=2$ and $d_H(v_j) \ge 2$ (by Claim~\ref{claim:girth6}.\ref{claim1v}). By contradiction we assume that each $v_i$, for $i \in \llbracket k-3 \rrbracket$, is either a $3^-$-vertex or a $4_3$-vertex. Then let $H' = H - \{uu_1, uu_2,..., uu_{k-3}, u_1v_1, u_2v_2,...,u_{k-3}v_{k-3} \}$. In order to extend $\phi$ to $H$, we distinguish the following two cases:

\begin{itemize}
\item Assume that there exists a vertex $v_i$ with $i\in \llbracket k-3 \rrbracket$ such that $v_i$ is a $3^-$-vertex. Without loss of generality assume this vertex is $v_{k-3}$. Then we colour each edge $uu_i$ for $i \in\llbracket k-4 \rrbracket$ (this is possible since $|L \setminus SC_\phi(N_2(uu_i))| \ge \Delta-4 \ge k-4$). We continue by colouring $uu_{k-3}$ and $u_{k-3}v_{k-3}$ in this order, so that at each step there is at least one colour left. Now, for each edge $u_iv_i$ with $i\in \llbracket k-4 \rrbracket$ we have $|L \setminus SC_\phi(N_2(u_iv_i))| \ge 1$ and we can colour them independently.
\item Each vertex $v_i$, with $i\in \llbracket k-3 \rrbracket$, is a $4_3$-vertex. Let $v$ be a 2-vertex adjacent to $v_{k-3}$ and distinct from $u_{k-3}$. We uncolour the edge $vv_{k-3}$. Now, similarly to the previous case, we colour each edge $uu_i$ for $i \in\llbracket k-4 \rrbracket$ and this is possible since for all $i$, $|L \setminus SC_\phi(N_2(uu_i))| \ge \Delta-4 \ge k-4$. Now, we colour $uu_{k-3}$, $vv_{k-3}$ and $u_{k-3}v_{k-3}$ in this order (at each step we have at least one colour left for the current edge). It remains to colour the edges $u_iv_i$, with $i\in \llbracket k-4 \rrbracket$, and since $|L \setminus SC_\phi(N_2(u_iv_i))| \ge 1$ we can colour them independently.
\end{itemize}

\item Let $u$ be a $k$-vertex in $H$ with $5 \le k \le \Delta$ and suppose by contradiction that $u$ is adjacent to $\alpha$ 1-vertices and to $k-2-\alpha$ 2-vertices $u_1,\ldots,u_{k-2-\alpha}$, such that for each $i \in \llbracket k-3-\alpha\rrbracket$, $u_i$ is adjacent to either a $3^-$-vertex or a $4_3$-vertex $v_i$.

Let $H' = H - \{uu_1,uu_2,...,uu_{k-3-\alpha},u_1v_1,u_2v_2,...,u_{k-3-\alpha}v_{k-3-\alpha}\}$. Then we proceed exactly as in the proof of the previous claim.
\end{enumerate}

\end{proof}

\subsection{Discharging procedure}

Euler's formula $|V(H)|-|E(H)|+|F(H)|=2$ can be rewritten as $(4|E(H)|-6|V(H)|)+(2|E(H)|-6|F(H)|)=-12$. Using the relation $\displaystyle \sum_{v \in V(H)}d(v)=\sum_{f \in F(H)}r(f)=2|E(H)|$ we get that:

\begin{equation}\label{eqGirth6}
\sum_{v \in V(G)}\,(2d(v)-6)\,+\,\sum_{f \in F(G)}\,(r(f)-6)\,=\,-12
\end{equation}

We define the weight function $\omega\,:\,V(H)\,\cup\,F(H)\,\longrightarrow\,\mathbb{R}$ by $\omega(x)=2d(x)-6$ if $x\,\in\,V(H)$ and $\omega(x)=r(x)-6$ if $x\,\in\,F(H)$. It follows
from Equation (\ref{eqGirth6}) that the total sum of weights is equal to
-12. In what follows, we will define discharging rules (R1) to (R6)
and redistribute weights accordingly. Once the discharging is
finished, a new weight function $\omega^\ast$ is produced. However,
the total sum of weights is kept fixed when the discharging is finished. Nevertheless, we will show that $\omega^\ast(x) \ge 0$ for all
$x\,\in\,V(H)\,\cup\,F(H)$. This will lead us to the following
contradiction:
$$0\; \le \sum_{x\,\in\,V(H)\,\cup\,F(H)}\;\omega^\ast(x)\;=\;\sum_{x\,\in\,V(H)\,\cup\,F(H)}\;\omega(x)\;=\;-12\;<\;0$$
and hence will demonstrate that such a counterexample cannot exist.

\medskip

The discharging rules are defined as follows:

\begin{enumerate}
\item[(R1)] Every face gives 2 to each incident 1-vertex.

\item[(R2)] Every $k$-vertex, for $k\geq 5$, gives 2 to each adjacent 1-vertex.

\item[(R3)] Every $4_3$-vertex gives $\frac{2}{3}$ to each adjacent 2-vertex.

\item[(R4)] Every $4_2$-vertex gives 1 to each adjacent 2-vertex.

\item[(R5)] Every $4_1$-vertex gives 2 to the adjacent 2-vertex.

\item[(R6)] Every $k$-vertex, for $k\geq 5$, gives:
\begin{enumerate}
\item[(R6.1)] 2 to each adjacent 2-vertex if this 2-vertex is adjacent to a $3^-$-vertex.
\item[(R6.2)] $\frac{4}{3}$ to each adjacent 2-vertex if this 2-vertex is adjacent to a $4_3$-vertex.
\item[(R6.3)] 1 to each adjacent 2-vertex if this 2-vertex is adjacent to a $4^+$-vertex distinct from a $4_3$-vertex.
\end{enumerate}
\end{enumerate}

\medskip

\noindent
Let $v\,\in\,V(H)$ be a $k$-vertex. Consider the following cases:

\begin{enumerate}
\item[] {\bf Case $\boldsymbol{k=1.}$} Observe that $\omega(v)=-4$.
By Claim~\ref{claim:girth6}.\ref{claim1v}, $v$ is adjacent to a $5^+$-vertex. By (R1) $v$ receives 2 from its incident face and by (R2) $v$ receives 2 from its adjacent vertex. Hence, $\omega^\ast(v)=-4+2+2 = 0$.
\item[] {\bf Case $\boldsymbol{k=2.}$} Observe that $\omega(v)=-2$. By Claim~\ref{claim:girth6}.\ref{claim1v}, $v$ has two neighbours $u$ and $w$ both of degree at least 2. Consider the following cases:
\begin{enumerate}
\item Suppose one of the neighbours of $v$, say $u$, is a $3^-$-vertex. Then by Claim~\ref{claim:girth6}.\ref{claim2v}, $w$ is a $4^+$-vertex. If $d(w)=4$ then by Claim~\ref{claim:girth6}.\ref{claim423v}, $w$ is a $4_1$-vertex and by (R5) we have $\omega^\ast(v)=-2+2= 0$. If $d(w) \ge 5$ then by (R6.1) we have $\omega^\ast(v)=-2+2=0$.

\item Assume now that $d(u)=d(w)=4$. Suppose first that $w$ is a $4_1$-vertex. Then by (R3), (R4) and (R5), $\omega^\ast(v)=-2+1\times 2+1\times \min \{2,1,\frac{2}{3}\} \geq 0$. Assume now that $w$ is a $4_2$-vertex. Then by Claim~\ref{claim:girth6}.\ref{claim4342v}, $u$ is not a $4_3$-vertex. Hence, by (R4) and (R5), $\omega^\ast(v)=-2+\min \{1 \times 1 + 1 \times 2,2 \times 1\} \ge 0$. Finally, suppose that $w$ is a $4_3$-vertex. Then by Claim~\ref{claim:girth6}.\ref{claim4342v}, $u$ is a $4_1$-vertex. Hence, by (R3) and (R5), $\omega^\ast(v)=-2+1 \times \frac{2}{3} + 1 \times 2 > 0$.

\item Suppose $d(u) \ge 5$ and $d(w)=4$ (the case when $d(u)=4$ and $d(w) \ge 5$ is symmetric). If $w$ is a $4_1$-vertex then by (R5) and (R6.3), $\omega^\ast(v)=-2+1 \times 2+1 \times 1 > 0$. Assume now, $w$ is a $4_2$-vertex hence, by (R4) and (R6.3), $\omega^\ast(v)=-2+1 \times 1 + 1 \times 1 = 0$. Suppose now, $w$ is a $4_3$-vertex then by (R3) and (R6.2), $\omega^\ast(v)=-2+1 \times \frac{2}{3} + 1 \times \frac{4}{3} = 0$.

\item Assume $d(u) \ge 5$ and $d(w) \ge 5$. Hence, by (R6.3), $\omega^\ast(v)=-2+2 \times 1 = 0$.
\end{enumerate}

\item[] {\bf Case $\boldsymbol{k=3.}$} The initial charge of $v$ is $\omega(v)=0$ and it remains unchanged during the discharging process. Hence $\omega(v)=\omega^\ast(v) = 0$.

\item[] {\bf Case $\boldsymbol{k=4.}$} Observe that $\omega(v)=2$.
By Claim~\ref{claim:girth6}.\ref{claim1v} $v$ is not adjacent to a $1$-vertex. By Claim~\ref{claim:girth6}.\ref{claimk2v}, $v$ is adjacent to at most three $2$-vertices. If $v$ is a $4_1$-vertex, then by (R5), $\omega^\ast(v)=2-1 \times 2 = 0$. If $v$ is a $4_2$-vertex, then by (R4), $\omega^\ast(v)=2-2 \times 1 = 0$. Suppose now $v$ is a $4_3$-vertex. Hence, by (R3), $\omega^\ast(v)=2-3 \times \frac{2}{3} = 0$.

\item[] {\bf Case $\boldsymbol{k \ge 5.}$} Observe that $\omega(v)=2k-6$. Consider the following cases:
\begin{enumerate}
\item Assume $v$ is not adjacent to a $1$-vertex. By Claim~\ref{claim:girth6}.\ref{claimk2v}, $v$ is adjacent to at most $k-1$ $2$-vertices. If $v$ is adjacent to at most $k-3$ $2$-vertices then by (R6), $\omega^\ast(v) \ge 2k-6-(k-3)\times 2 = 0$. If the number of 2-neighbours of $v$ is $k-2$, then by Claim~\ref{claim:girth6}.\ref{claimk22v} at most $k-4$ of them have a $3^-$-neighbour or a $4_3$-neighbour. Hence, by (R6.1) and (R6.3), $\omega^\ast(v) \geq 2k-6-(k-4) \times 2 -2 \times 1 = 0$. Suppose now that the number of 2-neighbours of $v$ is exactly $k-1$. Then by Claim~\ref{claim:girth6}.\ref{claimk12v}, none of these $2$-neighbours is adjacent to a $3^-$-vertex or to a $4_3$-vertex. Therefore, by (R6.3), we have $\omega^\ast(v) = 2k-6-(k-1) \times 1 = k-5 \ge 0$.

\item Suppose $v$ is adjacent to $\alpha$ 1-vertices with $\alpha \geq 1$. By Claim~\ref{claim:girth6}.\ref{claimk31v} we have $\alpha\leq k-3$. Moreover, if $\alpha=k-3$ then, by the same claim, $v$ cannot be adjacent to a $2$-vertex and thus, by (R2), $\omega^\ast(v)=2k-6-(k-3)\times 2 = 0$. So we may suppose that $\alpha\leq k-4$.  If the number of 2-neighbours of $v$ is at most $k-3-\alpha$, then by (R6) $\omega^\ast(v)\geq 2k-6-\alpha\times 2-(k-3-\alpha)\times 2 = 0$. Suppose the number of 2-neighbours of $v$ is at least $k-2-\alpha$. This number cannot be $k-\alpha$ according to Claim~\ref{claim:girth6}.\ref{claimk2v}, and, since $\alpha\geq 1$, by Claim~\ref{claim:girth6}.\ref{claimk12v} this number cannot be $k-1-\alpha$ neither. So $v$ has exactly $k-2-\alpha$ neighbours of degree 2. Then by Claim~\ref{claim:girth6}.\ref{claima1v}, at most $k-4-\alpha$ of the $2$-neighbours of $v$ are adjacent to either a $3^-$-vertex or a $4_3$-vertex. Therefore, by (R2) and (R6), $\omega^\ast(v)\geq 2k-6-\alpha\times 2-(k-4-\alpha)\times 2-2 \times 1 = 0$.
\end{enumerate}

\end{enumerate}

Let $f \in F(H)$ be a $k$-face. By hypothesis on the girth condition we know that $k \ge 6$. Note that if $f$ has $\alpha$ incident $1$-vertices, then $k\geq 6+2\alpha$. Since $\omega(f)= k-6$, by (R1), $\omega^\ast(f) \geq k-6-2\alpha \ge 0$.

\medskip

After performing the discharging procedure the new weights of all faces and vertices are positive and therefore, $H$ cannot exist.

\section{Proof of Theorem \ref{thm:secondthm}}
\label{sec:girth5}

In this section we show how the proof of Theorem~\ref{thm:faudree_planar}, given by Faudree \emph{et al.} in~\cite{FGST90}, can be analysed in order to get a better bound for $\chi'_s$ for several subclasses of planar graphs. Below we provide this proof because its idea will be used to prove Theorem~\ref{thm:secondthm}.

\newproof{myproof}{Proof of Theorem \ref{thm:faudree_planar}}

\begin{myproof}
First, decompose the edges of the planar graph into $\Delta+1$ distinct matchings (this is possible by Vizing's Theorem). For each matching $M$ build the following graph $G_M$.

Each vertex of $G_M$ corresponds to an edge of $M$. Two vertices of $G_M$ are adjacent if the corresponding edges are adjacent in $G$ (do not form an induced matching in $G$). The graph $G_M$ is planar and hence its vertices can be coloured properly with 4 colours using the Four Colour Theorem. This colouring corresponds to a strong edge-colouring of the matching $M$ in $G$.

Since there are at most $\Delta +1$ matchings and for each we use 4 colours, we obtain a strong $4(\Delta+1)$-edge-colouring of $G$.
\qed
\end{myproof}

The two main tools used in the previous proof are Vizing's Theorem and the Four Colour Theorem. Specifically, the following holds:

\begin{remark}
\label{rem:faudree_planar}
For every graph $G$,

 $$\chi'_s(G)\leq \chi'(G) \cdot \chi(G_M)$$ where $\chi(G_M)$ is the chromatic number of the graph $G_M$.
\end{remark}

Therefore, if one could show that under some restrictions a planar graph is properly $\Delta$-edge-colourable, then the bound given by the proof of Faudree \emph{et al.} would be improved. To this end, we would like to mention the following conjecture:

\begin{conjecture}[Vizing's Planar Graph Conjecture~\cite{V65Thm}]
\label{conj:planar_vizing}
Every planar graph $G$ with $\Delta\geq 6$ satisfies $\chi'(G)=\Delta$.
\end{conjecture}

The cases of $\Delta\geq 7$ of this conjecture have been already shown:

\begin{theorem}[Vizing~\cite{V65Thm}]
\label{thm:Vizing_planar}
If $G$ is a planar graph with $\Delta \geq 8$ then $\chi'(G)=\Delta$.
\end{theorem}

\begin{theorem}[Sanders \& Zhao~\cite{SZ01} and Zhang~\cite{Z00}]
\label{thm:Sanders_Zhao}
If $G$ is a planar graph with $\Delta = 7$ then $\chi'(G)=\Delta$.
\end{theorem}

Next theorem treats the cases when $\Delta\in\{4,5\}$:

\begin{theorem}[Li and Luo~\cite{LL03}]
\label{thm:LiLuo}
Let $G$ be a planar graph with girth $g$. Then the following holds:

\begin{itemize}
\item if $g\geq 5$ and $\Delta\geq 4$, then $\chi'(G)=\Delta$,
\item if $g\geq 4$ and $\Delta\geq 5$, then $\chi'(G)=\Delta$.
\end{itemize}
\end{theorem}

Theorems~\ref{thm:Vizing_planar}, \ref{thm:Sanders_Zhao} and \ref{thm:LiLuo} imply the following.

\begin{corollary}
\label{cor:planar_class1}
Let $G$ be a planar graph with girth $g$. If $G$ satisfies one of the following conditions below, then $\chi'(G)=\Delta$
\begin{itemize}
\item $\Delta \geq 7$
\item $\Delta \geq 5$ and $g \geq 4$
\item $g \geq 4$
\end{itemize}
\end{corollary}

Now, from Corollary~\ref{cor:planar_class1} and Remark~\ref{rem:faudree_planar} we deduce Theorem~\ref{thm:secondthm}.

\bigskip

We summarize all the observations of this section in the following table, where the upper bounds marked in bold are the ones given by Theorem~\ref{thm:faudree_planar} and that have not been improved since then:

\begin{table}[!ht]
\centering
\begin{tabular}{c|c|c|c|c}
 & $\Delta\geq 7$ & $\Delta\in\{5,6\}$ & $\Delta=4$ & $\Delta=3$\\ \hline
no girth restriction & $4\Delta$ & $\mathbf{4\Delta+4}$ & $\mathbf{4\Delta+4}$ & $3\Delta+1$\\ \hline
$g\geq 4$ & $4\Delta$ & $4\Delta$ & $\mathbf{4\Delta+4}$ & $3\Delta+1$\\ \hline
$g\geq 5$ & $4\Delta$ & $4\Delta$ & $4\Delta$ & $3\Delta+1$\\ \hline
$g\geq 6$ & $3\Delta+1$ & $3\Delta+1$ & $3\Delta+1$ & $3\Delta$\\ \hline
$g\geq 7$ & $3\Delta$ & $3\Delta$ & $3\Delta$ & $3\Delta$\\ \hline
\end{tabular}
\caption{Known upper bounds on the strong chromatic index of discussed subclasses of planar graphs}
\end{table}

The last line of the table is an immediate consequence of Gr\"otzsch's Theorem and Theorem~\ref{thm:faudree_planar} as observed in~\cite{CMPR13} and~\cite{HLSS13}. The last column is given by the results from~\cite{A92,HQT93} and Theorem~\ref{thm:subcubic}.

\bigskip
\section{Concluding remarks and open problems}
As mentioned in the introduction, for each $\Delta\geq 4$ there exist a planar graph $G$ of girth 4 such that $\chi'_s(G)=4\Delta-4$~\cite{FGST90}. Thus, the values in the first three rows of the table might not be optimal.
For planar graphs of girth 6 there exists graphs satisfying
$\chi'_s(G) \geq \Bigg \lceil \frac{12(\Delta-1)}{5} \Bigg \rceil$ as shown
by Hud\'{a}k \emph{et al.} \cite{HLSS13}.

Regarding Conjecture~\ref{conj:planar_vizing}, the condition of $\Delta\geq 6$ cannot be improved as Vizing himself showed in~\cite{V65Thm} that there exist planar graphs with $\Delta\in\{2,3,4,5\}$ and which are not properly $\Delta$-edge-colourable. The graphs proposed in his paper for $\Delta\in\{3,4,5\}$ are the graphs of geometric solids having one edge subdivided. For the cases of $\Delta\in\{4,5\}$ these graphs contain many triangles. Moreover, Theorem~\ref{thm:LiLuo} shows that planar graphs with $\Delta\geq 5$ having girth at least 4 are properly $\Delta$-edge-colourable and thus this result is tight (the size of the girth cannot be decreased). Therefore, the remaining natural question to which we could not find an answer is the following:

\begin{question}
Let $G$ be a planar graph with $\Delta = 4$ and girth at least 4. Is it true that $\chi'(G)=\Delta$?
\end{question}


\begin{thebibliography}{99}
\bibitem{A92}
L.D.~Andersen. The strong chromatic index of a cubic graph is at most 10, \emph{Discrete Math.} 108: 231--252, 1992.

\bibitem{BI13}O.V. Borodin and A.O. Ivanova. Precise upper bound for the strong edge chromatic number of sparse planar graphs, \emph{Discussiones Mathematicae Graph Theory} 33: 759--770, 2013.

\bibitem{CMPR13}G.J. Chang, M. Montassier, A. P\^echer, A. Raspaud. Strong chromatic index of planar graphs with large girth, \emph{The Seventh European Conference on Combinatorics, Graph Theory and Applications}, CRM Series (16): 265--270, 2013.

\bibitem{E88}P. Erd\H{o}s. Problems and results in combinatorial analysis and graph theory, \emph{Discrete Math.} 72(1-3): 81--92, 1988.

\bibitem{EN89}P. Erd\H{o}s and J. Ne\v{s}et\v{r}il. Irregularities of partitions ({G}. {H}al{\'a}sz, {V}. {T}. {S}{\'o}s, {E}ds.), [Problem], 162--163, 1989.

\bibitem{FGST90}
R.J. Faudree, A. Gy\'arfas, R.H. Schelp and Zs. Tuza. The strong chromatic index of graphs, \emph{Ars Combinatoria} 29B: 205--211, 1990.

\bibitem{HMRV13}
H. Hocquard, M. Montassier, A. Raspaud and P. Valicov. On strong edge-colouring of subcubic
graphs, \emph{Discrete Applied Math.} 161(16-17): 2467--2479, 2013.

\bibitem{HQT93} P.~Hor{\'a}k, H.~Qing and W.T.~Trotter, Induced matchings in cubic graphs, \emph{J. of Graph Theory} 17: 151--160, 1993.

\bibitem{HLSS13}
H.~Hud{\'a}k, B. Lu{\v z}ar, R. Sot{\'a}k and R. {\v S}krekovski. Strong edge-coloring of planar graphs, \emph{Discrete Math.} 324: 41--49, 2014.

\bibitem{LL03}
X. Li and R. Luo, Edge Coloring of Embedded Graphs with Large Girth. \emph{Graphs and Combinatorics} 19(3): 393--401, 2003.

\bibitem{MR97}
M. Molloy and B. Reed. A bound on the strong chromatic index of a graph. \emph{J. Combin. Theory B} 69(2): 103--109, 1997.

\bibitem{SZ01}
D.P. Sanders and Y. Zhao. Planar graphs of maximum degree seven are Class I, \emph{J. Combin. Theory B} 83: 201--212, 2001.

\bibitem{V65Thm}
V.G. Vizing. Critical graphs with given chromatic index, \emph{Metody Diskret. Analiz.} 5: 9--17, 1965 (In Russian).

\bibitem{Z00}
L. Zhang. Every planar graph with maximum degree 7 is of class 1, \emph{Graphs and Combinatorics} 16: 467--495, 2000.

\end{thebibliography}
\end{document}